\documentclass[a4paper,pra,twocolumn,groupedaddress]{revtex4}

\usepackage{ae}
\usepackage{bm}
\usepackage{physics}
\usepackage{anyfontsize}
\usepackage[T1]{fontenc}
\usepackage{mathptmx} 
\usepackage[cal=cm]{mathalfa}
\usepackage[ansinew]{inputenc}
\usepackage{amsmath}
\usepackage{amsthm}
\usepackage{amssymb}
\usepackage{stmaryrd}
\usepackage{amsmath,amscd,amsfonts,amssymb, color,bbm, bm,
	braket, xcolor, setspace, cancel, stmaryrd, calc, amsthm} 
\usepackage[caption=false]{subfig}
\usepackage{multirow}
\usepackage{array}
\usepackage[]{graphicx}
\usepackage{wrapfig}
\usepackage{makecell}
\usepackage{xparse}
\usepackage{dcolumn}
\usepackage{color}
\usepackage[colorlinks=true, linkcolor=red, citecolor=blue, urlcolor=blue]{hyperref}
\usepackage{lscape}
\newtheorem{theorem}{Theorem}
\graphicspath{ {./images1/} }

\newtheorem{prop}{Proposition}
\newtheorem{defin}{Definition}

\begin{document}
	\title{Hierarchy of hidden nonlocality: A genuine activation of Incompletability}
	\author{Soumajit Das}
	\email{soumajitd01@gmail.com}
	\affiliation{Physics and Applied Mathematics Unit, Indian Statistical Institute,  203 B.T. Road, Kolkata 700108, India}
	\author{Shampa Mondal}
	\email{shampamondal2451@gmail.com}
	\affiliation{Physics and Applied Mathematics Unit, Indian Statistical Institute,  203 B.T. Road, Kolkata 700108, India}
	\author{Preeti Parashar}
	\email{parashar@isical.ac.in}
	\affiliation{Physics and Applied Mathematics Unit, Indian Statistical Institute,  203 B.T. Road, Kolkata 700108, India}
	\author{Atanu Bhunia}
	\email{atanu.bhunia31@gmail.com}
	\affiliation{Physics and Applied Mathematics Unit, Indian Statistical Institute,  203 B.T. Road, Kolkata 700108, India}
\begin{abstract}
Quantum nonlocality admits several operational manifestations, one of which emerges from sets of orthogonal quantum states that cannot be perfectly distinguished by local operations and classical communication (LOCC). Such sets are regarded as nonlocal because their perfect discrimination requires global measurements. In contrast, sets that are perfectly distinguishable by LOCC are generally considered locally accessible and operationally classical.
In this work, we investigate the role of incompletability in local state discrimination and introduce the notion of \emph{activation of incompletability}. Specifically, we demonstrate the existence of orthogonal sets that are initially perfectly distinguishable by LOCC and free from local redundancy, but which can be transformed via LOCC into strictly incompletable sets. We prove that activation of incompletability necessarily implies activation of nonlocality, whereas the converse fails in general, thereby establishing a hierarchy between the two activation phenomena. Furthermore, within the framework of local incoherent operations and classical communication (LICC), we show that any set whose incompletability can be activated can nevertheless be extended to a complete orthonormal basis of the Hilbert space, although the resulting completed basis is no longer perfectly distinguishable by LOCC. Our results uncover a fundamental interplay among local distinguishability, incompletability, coherence, and nonlocality, and provide new insight into the structure of locally accessible quantum information.
\end{abstract}
	
	\date{\today}
	\maketitle
	\section{Introduction}
	\par Quantum systems exhibit several forms of nonlocality that are distinct from the well-known Bell nonlocality \cite{Bell nonlocality}. When a set of orthogonal quantum states cannot be perfectly distinguished by local operations and classical communication (LOCC), it reflects a fundamental nonlocal feature of quantum physics, \cite{BennettPB1999}. Local distinguishability of quantum states refers to the task of identifying a state from a set of prespecified orthogonal states shared among parties separated by arbitrary distances and LOCC being the only allowed class of operations \cite{BennettUPB1999,Walgate2000,Virmani,Ghosh2001,Groisman,Walgate2002,Divincinzo,Horodecki2003,Fan2004,Ghosh2004,Nathanson2005,Watrous2005,Niset2006,Ye2007,Fan2007,Runyo2007,somsubhro2009,Feng2009,Runyo2010,Yu2012,Yang2013,Zhang2014,somsubhro2010,yu2014,somsubhro2014,somsubhro2016,bennett1996,popescu2001,xin2008,somsubhro2009(1)}. The nonlocality of orthogonal quantum states can be used for various practical purposes such as data hiding \cite{terhal58,divincenzo580,lamidatahiding,terhaldatahiding,chaves2020,wehner2020,winterdatahiding,haydendatahiding}, quantum secret sharing \cite{rahaman330,markham309,wang320}, and similar applications. Consequently, 
considerable attention has been paid to the study of local distinguishability of orthogonal quantum states and the exploration of the relationship between quantum nonlocality and entanglement \cite{Zhang2015,Wang2015,Chen2015,Yang2015,Zhang2016,Xu2016(2),Zhang2016(1),Xu2016(1),Halder2019strong nonlocality,Halder2019peres set,Xzhang2017,Xu2017,Wang2017,Cohen2008,somsubhro2018,zhang2018,Halder2018,Yuan2020,Rout2019,bhunia2020,bhunia2023,biswas2023,Zhang2019,bhunia2022}.

\par The operational framework underlying local distinguishability is precisely the class of LOCC. In a typical multipartite setting, the subsystems are distributed among distant parties who perform local measurements on their respective systems while communicating their outcomes classically. Based on the communicated information, subsequent local measurements are adaptively chosen, and the protocol proceeds either for a finite number of rounds or until the desired task is accomplished.

\par This class of operations is known as LOCC. From both practical and foundational perspectives, LOCC plays a central role in quantum information theory. Experimentally, local measurements are easier to implement and require fewer resources than global operations on composite systems. Fundamentally, LOCC is intimately connected to entanglement theory, since entanglement characterizes precisely those multipartite correlations that cannot be generated through LOCC alone \cite{quantum entanglement}. Despite its broad significance, however, the structure of LOCC remains far from completely understood.
	\par Local distinguishability of quantum states provides a fundamental framework for understanding the operational limitations of LOCC. In a seminal work, Walgate \textit{et al.}~\cite{Walgate2000} showed that any two orthogonal multipartite pure states can always be perfectly distinguished by LOCC. However, this property no longer holds in general for larger sets of orthogonal states. Indeed, the existence of locally indistinguishable orthogonal states reveals a distinctive manifestation of quantum nonlocality.

\par Since entanglement is deeply connected to nonlocal phenomena, it is natural to expect that sets of mutually orthogonal product states should always be perfectly distinguishable by LOCC. Surprisingly, this intuition turns out to be incorrect. Bennett \textit{et al.}~\cite{BennettPB1999} first demonstrated this by constructing a set of nine orthogonal product states in $\mathbb{C}^3 \otimes \mathbb{C}^3$ that cannot be perfectly distinguished by LOCC, thereby introducing the phenomenon of ``nonlocality without entanglement.'' This remarkable result established that entanglement is not a necessary resource for local indistinguishability and showed that the absence of entanglement alone does not guarantee local accessibility of information. Since then, numerous examples of locally indistinguishable product states and related structures have been discovered \cite{BennettPB1999,BennettUPB1999,Zhang2015,Wang2015,Chen2015,Yang2015,Zhang2016,Xu2016(2),Zhang2016(1),Xu2016(1),bhunia2023,biswas2023,Halder2019strong nonlocality,Halder2019peres set,Xzhang2017,Xu2017,Wang2017,Cohen2008,Zhang2019,somsubhro2018,zhang2018,Halder2018,Yuan2020,Rout2019,bhunia2020,bhunia2022,bhuniaubb2024,indra2025,subrata2024,bhunia2025}.
\par A particularly important restricted subclass of LOCC is given by local incoherent operations and classical communication (LICC) \cite{streltsov2015,streltsov2017,chakraborty2019}, where each local operation is required to be incoherent with respect to a fixed reference basis. Since incoherent operations cannot create quantum coherence from incoherent states \cite{T2014,W2016,S2017}, the LICC framework naturally incorporates coherence-theoretic constraints into local state discrimination protocols. Quantum coherence, being a basis-dependent manifestation of superposition, is now recognized as a fundamental quantum resource with wide-ranging applications in quantum information processing. Consequently, studying local distinguishability under LICC provides a natural platform for exploring the interplay among coherence, entanglement, and quantum nonlocality \cite{A2021}.
\par A particularly important class of such sets is the unextendible product basis (UPB), which provides an incomplete orthogonal product basis exhibiting nonlocality without entanglement. A UPB is a set of mutually orthogonal product states whose complementary subspace contains no product state \cite{Divincinzo}. Equivalently, the set cannot be extended to a complete product basis while preserving orthogonality. More generally, a set of orthogonal product states is said to be \emph{incompletable} if it cannot be completed to a full orthogonal product basis of the underlying Hilbert space, while a set that admits no such completion even in any enlarged Hilbert space is referred to as \emph{strongly incompletable}. It is known that UPBs cannot be perfectly distinguished by LOCC \cite{LOCC}, and the normalized projector onto their orthogonal complement gives rise to bound entangled states \cite{BennettUPB1999,Divincinzo}. Consequently, UPBs establish a profound connection between local indistinguishability, incompletability, and entanglement, making them central objects in quantum information theory.

    \par These observations naturally raise a broader question concerning the behavior of completable orthogonal sets under basis extension. Consider two orthogonal sets of pure states, $\mathcal{S}_A$ and $\mathcal{S}_B$ in the same Hilbert space. Suppose that $\mathcal{S}_A$ contains at least one entangled state together with several product states, whereas $\mathcal{S}_B$ consists entirely of product states. Assume further that both sets are LOCC distinguishable and completable to orthonormal bases of the full Hilbert space. By the seminal result of Horodecki \textit{et al.}\cite{Horodecki2003}, the completed basis extending $\mathcal{S}_A$ is necessarily locally indistinguishable, since every orthonormal basis containing at least one entangled state cannot be perfectly distinguished by LOCC. In contrast, what can be concluded about the LOCC distinguishability of a completion of $\mathcal{S}_B $ consisting entirely of product states?  Motivated by this  question, we investigate here the role of incompletability in local state discrimination and introduce the notion of \emph{activation of incompletability}. In particular, we demonstrate that there exist sets of orthogonal states which are initially perfectly distinguishable by LOCC and free from local redundancy \cite{bandhyopadhyay201}, yet can be transformed through LOCC into strictly incompletable sets. We show that activation of incompletability constitutes a strictly stronger phenomenon than activation of nonlocality \cite{bandhyopadhyay201,GhoshStrongActivation2022,Li2022,subrata2024,bera2026,bhunia2026}: every set exhibiting activation of incompletability necessarily exhibits activation of nonlocality, whereas the converse does not hold in general. This establishes a nontrivial hierarchy between the two activation phenomena.
    
\par Furthermore, by restricting the allowed operations to LICC, we uncover a fundamental connection among incompletability, coherence, and local distinguishability. Specifically, we prove that any set whose incompletability can be activated under LICC can be extended to a complete orthonormal basis of the underlying Hilbert space, although the resulting completed basis fails to remain perfectly distinguishable by LOCC. Our results, therefore, provide a new operational perspective on the interplay between coherence and nonlocality and shed further light on the structure of locally accessible quantum information.
The remainder of the paper is organized as follows. In Sec.~\ref{A1}, we introduce the necessary definitions and preliminary concepts related to local distinguishability, incompletability, coherence, and activation phenomena. Explicit examples illustrating activation of nonlocality are presented in Sec.~\ref{A2}. We define in Sec.~\ref{A3}, activation of incompletability and analyze its relation to activation of nonlocality in both bipartite and multipartite settings. In Sec.~\ref{A4}, we investigate the role of LICC and establish the fundamental connections among incompletability, coherence, and local distinguishability. Finally, we summarize our results and discuss possible future directions in Sec.~\ref{A5}.	
	\section{Preliminaries}
	\label{A1}
	A measurement on a $d$-dimensional quantum system can be expressed as a set of positive operator-valued measure (POVM) elements $\left\{M_k\right\}_k$. These elements are the positive semidefinite Hermitian matrices that satisfy the completeness relation $\sum_k M_k=\mathrm{I}_d$, where $\mathrm{I}_d$ is the identity matrix of order $d$. In this section, we will first review some of the definitions which are used throughout the following sections.\\
	\begin{defin}
		\cite{Walgate2002, Halder2018} If all the POVM elements of a measurement structure, corresponding to a discrimination task of a given set of states, are proportional to the identity matrix, then such a measurement is not useful to extract information for this task and is called a $trivial\;measurement$. On the other hand, if not all the POVM elements of a measurement
		are proportional to the identity matrix, then the measurement is said to be a \emph{nontrivial measurement}.
	\end{defin}
	\begin{defin} \cite{Walgate2002, Halder2018} Consider a local measurement to distinguish a fixed set of pairwise orthogonal quantum states. After performing that measurement, if the post-measurement states remain	pairwise orthogonal to each other then such a measurement is said to be an $orthogonality-preserving\; local\;measurement$ (OPLM).
	\end{defin}
	In this work, we shall always stick to OPLM. 
	\begin{defin}
		\cite{Halder2019strong nonlocality} A set of orthogonal quantum states is $locally\; irreducible$ if it is not possible to eliminate one or more quantum states from the set by nontrivial OPLM.
	\end{defin}
	\begin{defin} 
		A set of orthogonal
		quantum states is $locally\; indistinguishable$  if it is possible to eliminate one or more states from the set by OPLM, but it is not possible to distinguish the whole set completely by nontrivial OPLM.
	\end{defin}
\begin{figure}[t]
		\centering
        \hspace*{-1cm}
		\includegraphics[width=1.2\linewidth]{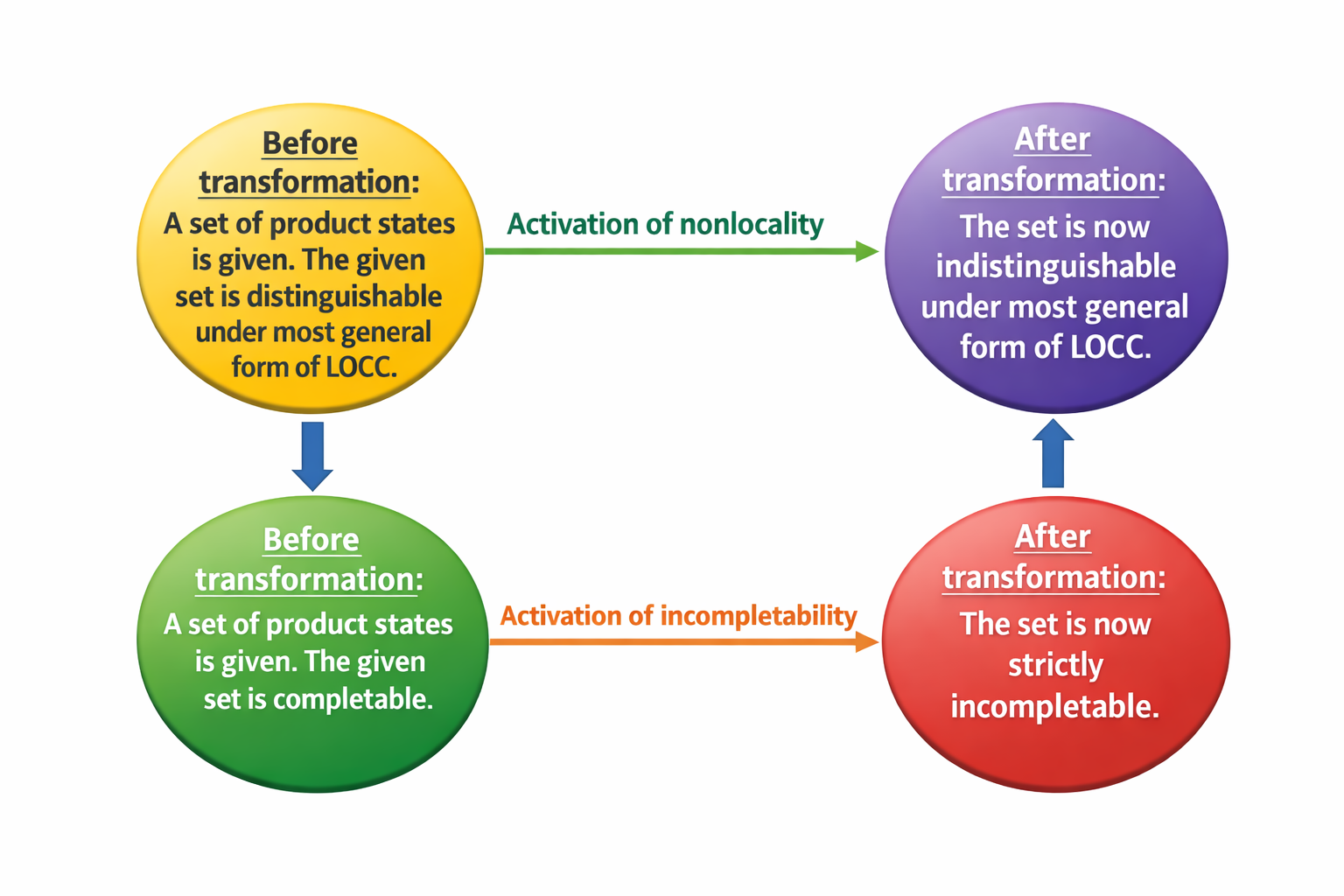}
		\caption{Schematic illustration of the proposed transformation. A LOCC distinguishable, completable set of product state is transformed into a strictly incompletable, LOCC-indistinguishable set, simultaneously activating incompletability and nonlocality.}
		\label{fig1}
	\end{figure}
    
	\par Therefore, it is by definition implied that all locally irreducible states are locally indistinguishable, but the converse is not true.   
	\begin{defin}
    Consider a multipartite set of mutually orthogonal quantum states $S=\left\{|\psi_i\rangle\right\}\subseteq \mathcal{H}_1\otimes\mathcal{H}_2\otimes\cdots\otimes\mathcal{H}_n$. The set S is said to possess local redundancy if there exists a proper subsystem $(A\subsetneq {1,2,\ldots,n})$ such that, after discarding the complementary subsystem $\bar A$, the resulting reduced states remain mutually orthogonal. 

If no such subsystem exists, i.e., if the pairwise orthogonality of the states is destroyed under every nontrivial local reduction, then the set S is said to be locally irredundant (or free from local redundancy).
    \end{defin}
	\begin{defin}\cite{bandhyopadhyay201}
		A locally distinguishable set $\mathcal{S}$ of multipartite orthogonal states is said to be nonlocality activable if it can be transformed to a set of locally indistinguishable orthogonal states via OPLM.    
	\end{defin} 

    \begin{defin}
		A locally distinguishable set $\mathcal{S}$ of multipartite orthogonal states is said to be incompletability activable if it can be transformed to a set of locally incompletable orthogonal states via OPLM.    
	\end{defin} 
	\par As all incompletable sets of orthogonal states are indistinguishable by LOCC \cite{BennettUPB1999},  therefore, it is by definition implied that all incompletability activable sets are nonlocality activable, but the converse is not always true.
	\section{Activation of Nonlocality}
	\label{A2}
    In this section, we present explicit constructions illustrating the phenomenon of activation of nonlocality. We begin with a bipartite example that serves as a prototype for the activation mechanism under local measurements and classical communication. Let $\{|0\rangle, \mid 1\}, \ldots,|d-1\rangle\}$ be an orthonormal basis in $\mathbb{C}^d$, where $d \geqslant 2$. Then, rank-2 and rank-3 projection operators (projectors) are defined as
	
	$$
	\begin{aligned}
		P_{i j} & =|i\rangle(i|+| j)\langle j|,\quad i \neq j, \\
		P_{i j k} & =|i\rangle\langle i|+|j\rangle\langle j|+|k\rangle\langle k|,\quad i \neq j \neq k,
	\end{aligned}
	$$
	
	respectively, where $i, j, k \in\{0,1, \ldots, d-1\}$. For example, $P_{\mathrm{01}}=|0\rangle\langle 0|+|1\rangle\langle 1|$ and $P_{012}=|0\rangle\langle 0|+|1\rangle\langle 1|+|2\rangle\langle 2|$. Let us assume a bipartite system in $\mathbb{C}^2 \otimes \mathbb{C}^4$, where Alice possesses a qubit and Bob possesses two qubits. For notational convenience, we relabel Bob's computational basis as $|00\rangle \equiv |\mathbf{0}\rangle,|01\rangle \equiv|\mathbf{1}\rangle,|10\rangle \equiv|\mathbf{2}\rangle$, and $|11\rangle \equiv|\mathbf{3}\rangle$. Consider the bipartite set $\mathcal{S}_1 =\left\{\left|\xi_i\right\rangle_{A B}\right\}_{i=1}^{3}\in \mathbf{\mathbf{C}^2 \otimes \mathbf{C}^4}$ \cite{bandhyopadhyay201}, where
	
	\begin{multline}
	\begin{aligned}
		& \left|\xi_1\right\rangle =|0\mathbf{0}\rangle+|0\mathbf{2}\rangle+|1\mathbf{1}\rangle-|1\mathbf{3}\rangle, \\
		& \left|\xi_2\right\rangle =|0\mathbf{0}\rangle-|0\mathbf{2}\rangle-|1\mathbf{1}\rangle-|1\mathbf{3}\rangle, \\
		& \left|\xi_3\right\rangle =|0\mathbf{1}\rangle-|1\mathbf{2}\rangle-|1\mathbf{0}\rangle-|0\mathbf{3}\rangle .
	\end{aligned} 
    \label{1}
	\end{multline}
	It is straightforward to verify that the set is free from local redundancy. In particular, not all pairs remain orthogonal if we discard any of Bob's qubits. To show that the states in (\ref{1}) are locally distinguishable, Alice first  performs a measurement on her qubit in the $ \{\mid 0\rangle,|1\rangle\}$ basis and tells Bob the result. Now, each of Alice's outcome results in a set of three orthogonal states for Bob to distinguish. If Alice gets  \lq\lq0\rq\rq, Bob distinguishes the states $|\mathbf{0}\rangle \pm|\mathbf{2}\rangle$ and $|\mathbf{1}\rangle-|\mathbf{3}\rangle$, and if Alice gets \lq\lq1\rq\rq, Bob distinguishes $|\mathbf{1}\rangle \mp|\mathbf{3}\rangle$ and $|\mathbf{0}\rangle+|\mathbf{2}\rangle$.
	\begin{prop}\cite{bandhyopadhyay201}
		The set $\mathcal{S}_1$ is a locally distinguishable set and can be transformed deterministically to a locally indistinguishable set via orthogonality-preserving LOCC.
	\end{prop}
    \begin{proof}
	To prove this, let us first choose Bob to perform a binary measurement defined by the orthogonal projectors $P_{01}$ and $P_{23}$ and inform Alice of the outcome. If Bob applies $P_{01}$ they are left with one of $|0\mathbf{0}\rangle \pm|1\mathbf{1}\rangle$ and $|0\mathbf{1}\rangle-|1\mathbf{0}\rangle$. Or, if Bob applies $P_{23}$ they are left with one of $|0\mathbf{2}\rangle \mp|1\mathbf{3}\rangle$ and $|1\mathbf{2}\rangle+|0\mathbf{3}\rangle$. So, in each case, they are left with one of three mutually orthogonal pure entangled states that can be embedded in a $\mathbb{C}^2 \otimes \mathbb{C}^2$ space. But as we know, three Bell states in $\mathbb{C}^2 \otimes \mathbb{C}^2$ are always indistinguishable, so,  each set is locally indistinguishable \cite{Ghosh2001,Walgate2002}. Therefore, the states in (\ref{1}) can always be locally converted into another set of three orthogonal states that cannot be locally distinguished. This completes the proof.
     \end{proof}
    The above phenomenon is referred to as \emph{activation of nonlocality}. Operationally, one starts with a set of orthogonal states that is perfectly distinguishable by LOCC and, through orthogonality-preserving local measurements assisted by classical communication, transforms it into a set that is locally indistinguishable, and hence nonlocal \cite{BennettPB1999,Walgate2000}. This notion of activated nonlocality was introduced in \cite{bandhyopadhyay201}.
	\section{Activation of Incompletability}
    \label{A3}
	In this section, we investigate activation of incompletability and establish its relation to activation of nonlocality. In particular, we show that activation of incompletability constitutes a strictly stronger form of activation phenomenon, thereby inducing a hierarchy within locally distinguishable multipartite sets. To illustrate this structure explicitly, we begin with the following bipartite set $\mathcal{S}_2 =\left\{\left|\phi_i\right\rangle_{A B}\right\}_{i=1}^{10}\in \mathbf{\mathbf{C}^6 \otimes \mathbf{C}^6}$, where
	\begin{figure}[t]
		\centering
		\includegraphics[width=1.0\linewidth]{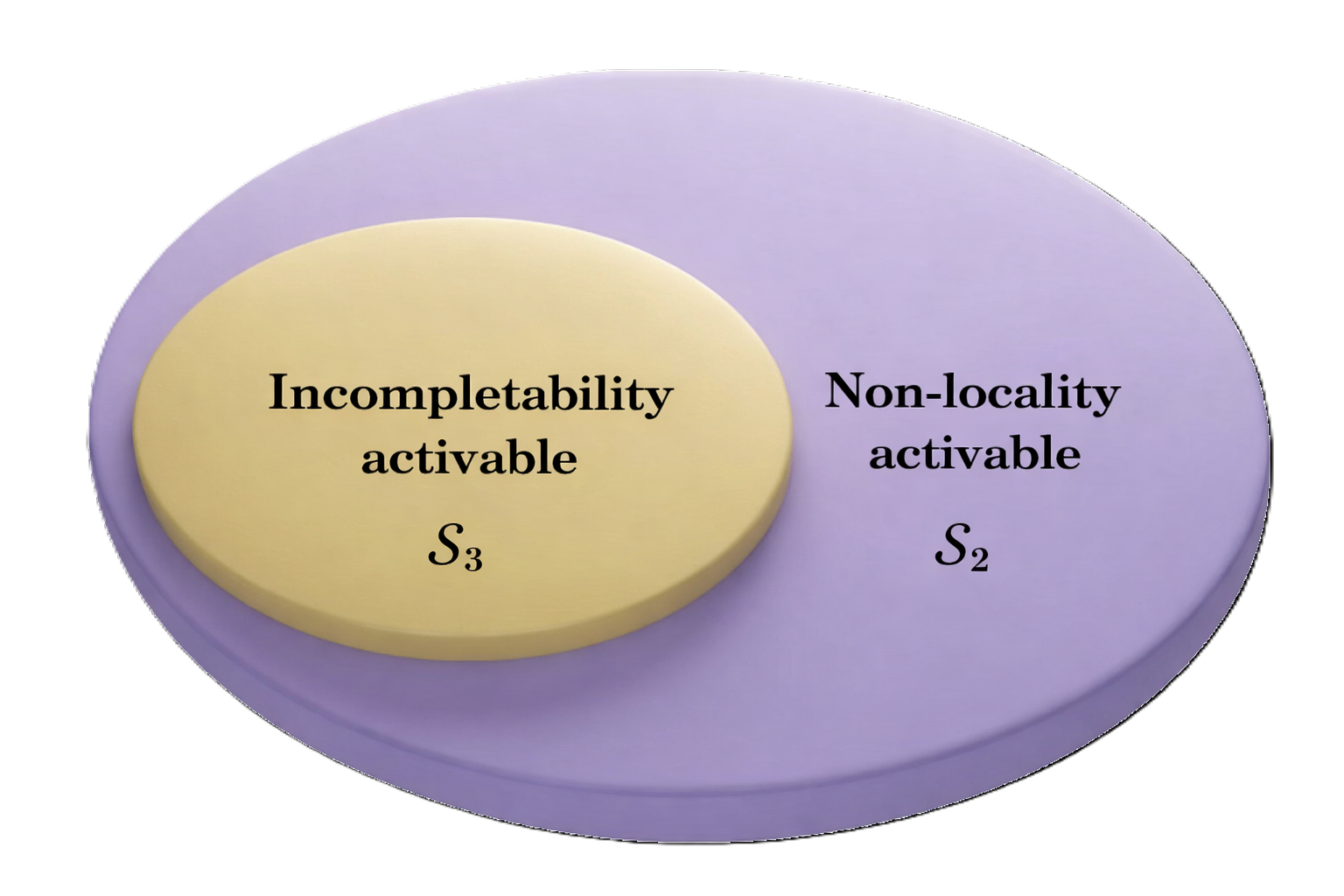}
		\caption{Hierarchy between activation of nonlocality and activation of incompletability. Every incompletability activable set is nonlocality activable, while the converse does not hold in general.}
		\label{fig1}
	\end{figure}
	\begin{multline}
		\begin{aligned}
			& \left|\phi_1\right\rangle_{A B}=|\mathbf{0}\rangle_A|\mathbf{0}-\mathbf{1}+\mathbf{4}-\mathbf{5}\rangle_B \\
			& \left|\phi_2\right\rangle_{A B}=|\mathbf{2}\rangle_A|\mathbf{1}-\mathbf{2}+\mathbf{5}-\mathbf{3}\rangle_B \\
			& \left|\phi_3\right\rangle_{A B}=|\mathbf{1-2}\rangle_A|\mathbf{0}-\mathbf{4}\rangle_B \\
			& \left|\phi_4\right\rangle_{A B}=|\mathbf{0-1}\rangle_A|\mathbf{2}-\mathbf{3}\rangle_B \\
			& \left|\phi_5\right\rangle_{A B}=|\mathbf{0+1+2}\rangle_A|\mathbf{0}+\mathbf{1}+\mathbf{2}+\mathbf{3}+\mathbf{4}+\mathbf{5}\rangle_B \\
			& \left|\phi_6\right\rangle_{A B}=|\mathbf{3}\rangle_A|\mathbf{0}-\mathbf{1}+\mathbf{4}-\mathbf{5}\rangle_B \\
			& \left|\phi_7\right\rangle_{A B}=|\mathbf{5}\rangle_A|\mathbf{1}-\mathbf{2}+\mathbf{5}-\mathbf{3}\rangle_B \\
			& \left|\phi_8\right\rangle_{A B}=|\mathbf{4-5}\rangle_A|\mathbf{0}-\mathbf{4}\rangle_B \\
			& \left|\phi_9\right\rangle_{A B}=|\mathbf{3-4}\rangle_A|\mathbf{2}-\mathbf{3}\rangle_B \\
			& \left|\phi_{10}\right\rangle_{A B}=|\mathbf{3+4+5}\rangle_A|\mathbf{0}+\mathbf{1}+\mathbf{2}+\mathbf{3}+\mathbf{4}+\mathbf{5}\rangle_B
		\end{aligned}
		\label{2}
	\end{multline}
	 The set $\mathcal{S}_2$ considered above can be seen to be free from local redundancy \cite{bandhyopadhyay201,Li2022,subrata2024}. Here, both Alice's and Bob's system can be considered to be the composition of qubit and qutrit subsystems. Precisely, $|\mathbf{0}\rangle:=|00\rangle,|\mathbf{1}\rangle:=|01\rangle,|\mathbf{2}\rangle:=$ $|02\rangle,|\mathbf{3}\rangle:=|10\rangle,\left|\mathbf{4}\right\rangle:=|11\rangle,|\mathbf{5}\rangle:=|12\rangle$. Take two states, $\left|\phi_3\right\rangle_{A B}$ and $\left|\phi_4\right\rangle_{A B}$. When any of the subparts (qubit or qutrit) of Bob's system for both states is discarded the reduced states will be nonorthogonal. An analogous argument applies to Alice's subsystem. This implies the set $\mathcal{S}_2$ is free from local redundancy.
	\par Now we will show that the set $\mathcal{S}_2$ is locally distinguishable. The players can avail the following discrimination protocol. First Bob performs a measurement $\mathcal{M}_B \equiv\left\{\mathcal{M}_B^1:=\right.$ $P\left[|\mathbf{0}-\mathbf{4}\rangle_B\right], \mathcal{M}_B^2 :=P\left[|\mathbf{2}-\mathbf{3}\rangle_B\right], \mathcal{M}_B^3:=P[\mid \mathbf{0}+\mathbf{1}+\mathbf{2}+\mathbf{3}+\mathbf{4}$ $\left.\left.+\mathbf{5}\rangle_B\right], \mathcal{M}_B^4:=\mathbb{I}-\left(\mathcal{M}_B^1+\mathcal{M}_B^2+\mathcal{M}_B^3\right)\right\}$. Here, $P\left[|\cdot\rangle\right]:=$ $|\cdot\rangle\langle\cdot|_{\mathcal{P}}$, and $\mathcal{P}$ denotes the party. When $\mathcal{M}_B^1$ clicks, the given state must be $\left|\phi_3\right\rangle$ and $\left|\phi_8\right\rangle$, which can be distinguished by Alice, projecting onto $|\mathbf{1-2}\rangle$ and $|\mathbf{4-5}\rangle$. Similarly, for the click $\mathcal{M}_B^2$, the states are $\left|\phi_4\right\rangle$ and $\left|\phi_9\right\rangle$, which can be distinguished by Alice, projecting onto $|\mathbf{0-1}\rangle$ and $|\mathbf{3-4}\rangle$. Also for the outcome $\mathcal{M}_B^3$ the isolated states are $\left|\phi_5\right\rangle$ and $\left|\phi_{10}\right\rangle$, which can be distinguished by Alice, projecting onto $|\mathbf{0+1+2}\rangle$ and $|\mathbf{3+4+5}\rangle$. Whenever $\mathcal{M}_B^4$ clicks the given state can be $\left|\phi_1\right\rangle$, $\left|\phi_2\right\rangle$, $\left|\phi_6\right\rangle$ and $\left|\phi_7\right\rangle$. However, in that case, Alice can perform a measurement $\mathcal{M}_A \equiv\left\{\mathcal{M}_A^1:=\right.$ $P\left[|\mathbf{0}\rangle_A\right], \mathcal{M}_A^2 :=P\left[|\mathbf{2}\rangle_A\right], \mathcal{M}_A^3:=P[|\mathbf{3}\rangle_A$ $\left], \mathcal{M}_A^4:=\mathbb{I}-\left(\mathcal{M}_A^1+\mathcal{M}_A^2+\mathcal{M}_A^3\right)\right\}$  to distinguish between these four states. This concludes the local discrimination protocol for the set $\mathcal{S}_2$.
    
The set in Eq.~(\ref{2}) consists entirely of product states in $\mathbf{\mathbf{C}^6 \otimes \mathbf{C}^6}$ and is therefore completable to a full orthonormal product basis \cite{BennettUPB1999}. We now show that this set can be transformed, via OPLM, into a strictly incompletable set. We refer to this phenomenon as \emph{activation of incompletability}. 

More precisely, activation of incompletability describes the process in which an initially completable and locally distinguishable set of orthogonal states is converted, under some LOCC, into a strictly incompletable set. In the following, we present an explicit protocol that realises such an activation for the set in Eq.~(\ref{2}).
    
	\begin{prop}
		The set $\mathcal{S}_2$ is initially a completable set and can be transformed deterministically to a strictly incompletable set via OPLM.    
	\end{prop}
	\begin{proof} 
		Consider that Bob performs a local measurement $\mathcal{K}_B \equiv\left\{\mathcal{K}^B_1:=P\left[(|\mathbf{0}\rangle,|\mathbf{1}\rangle,|\mathbf{2}\rangle)_B\right], \mathcal{K}^B_2:=\right.$ $\left.P\left[(|\mathbf{3}\rangle,|\mathbf{4}\rangle,|\mathbf{5}\rangle)_B\right]\right\}$, $P\left[(|i\rangle,|j\rangle,\dots)_{\mathcal{P}}\right] = \left[(|i\rangle\langle i|+|j\rangle\langle j|+\dots)_{\mathcal{P}}\right]$, $\mathcal{P}$ stands for party. If $\mathcal{K}^B_1$ clicks, they end up with
		
		$$
		\left\{\begin{array}{c}
			|\mathbf{0}\rangle_A|\mathbf{0}-\mathbf{1}\rangle_B,|\mathbf{2}\rangle_A|\mathbf{1}-\mathbf{2}\rangle_B, \\
			|\mathbf{1-2}\rangle_A|\mathbf{0}\rangle_B,|\mathbf{0-1}\rangle_A|\mathbf{2}\rangle_B, \\
			|\mathbf{0+1+2}\rangle_A|\mathbf{0}+\mathbf{1}+\mathbf{2}\rangle_B,\\
			|\mathbf{3}\rangle_A|\mathbf{0}-\mathbf{1}\rangle_B,|\mathbf{5}\rangle_A|\mathbf{1}-\mathbf{2}\rangle_B, \\
			|\mathbf{4-5}\rangle_A|\mathbf{0}\rangle_B,|\mathbf{3-4}\rangle_A|\mathbf{2}\rangle_B, \\
			|\mathbf{3+4+5}\rangle_A|\mathbf{0}+\mathbf{1}+\mathbf{2}\rangle_B
		\end{array}\right\}
		$$
		After that Alice makes measurement  $\mathcal{K}_A \equiv\left\{\mathcal{K}^A_1:=P\left[(|\mathbf{0}\rangle,|\mathbf{1}\rangle,|\mathbf{2}\rangle)_A\right], \mathcal{K}^A_2:=\right.$ $\left.P\left[(|\mathbf{3}\rangle,|\mathbf{4}\rangle,|\mathbf{5}\rangle)_A\right]\right\}$. If $\mathcal{K}^A_1$ occurs, the states reduce to
		\[
		\left\{\begin{array}{c}
			|\mathbf{0}\rangle_A|\mathbf{0}-\mathbf{1}\rangle_B,|\mathbf{2}\rangle_A|\mathbf{1}-\mathbf{2}\rangle_B, \\
			|\mathbf{1-2}\rangle_A|\mathbf{0}\rangle_B,|\mathbf{0-1}\rangle_A|\mathbf{2}\rangle_B, \\
			|\mathbf{0+1+2}\rangle_A|\mathbf{0}+\mathbf{1}+\mathbf{2}\rangle_B
		\end{array}\right\}
		\]
		which is a strictly incompletable set \cite{BennettUPB1999}. If $\mathcal{K}^A_2$ occurs, the resulting set is also strictly incompletable \cite{BennettUPB1999}
		\[ 
		\left\{\begin{array}{c}
			|\mathbf{3}\rangle_A|\mathbf{0}-\mathbf{1}\rangle_B,|\mathbf{5}\rangle_A|\mathbf{1}-\mathbf{2}\rangle_B, \\
			|\mathbf{4-5}\rangle_A|\mathbf{0}\rangle_B,|\mathbf{3-4}\rangle_A|\mathbf{2}\rangle_B, \\
			|\mathbf{3+4+5}\rangle_A|\mathbf{0}+\mathbf{1}+\mathbf{2}\rangle_B
		\end{array}\right\}
		\]
		On the other hand, if Bob gets $\mathcal{K}^B_2$, they are then left with the following states
		\[	\left\{\begin{array}{c}
			|\mathbf{0}\rangle_A|\mathbf{4}-\mathbf{5}\rangle_B,|\mathbf{2}\rangle_A|\mathbf{5}-\mathbf{3}\rangle_B, \\
			|\mathbf{1-2}\rangle_A|\mathbf{4}\rangle_B,|\mathbf{0-1}\rangle_A|\mathbf{3}\rangle_B, \\
			|\mathbf{0+1+2}\rangle_A|\mathbf{3}+\mathbf{4}+\mathbf{5}\rangle_B,\\
			|\mathbf{3}\rangle_A|\mathbf{4}-\mathbf{5}\rangle_B,|\mathbf{5}\rangle_A|\mathbf{5}-\mathbf{3}\rangle_B, \\
			|\mathbf{4-5}\rangle_A|\mathbf{4}\rangle_B,|\mathbf{3-4}\rangle_A|\mathbf{3}\rangle_B, \\
			|\mathbf{3+4+5}\rangle_A|\mathbf{3}+\mathbf{4}+\mathbf{5}\rangle_B
		\end{array}\right\}
		\]
		After that Alice makes measurement $\mathcal{K}^A_1$ or $\mathcal{K}^A_2$.  If $\mathcal{K}^A_1$ clicks, it gives
		\[ 
		\left\{\begin{array}{c}
			|\mathbf{0}\rangle_A|\mathbf{4}-\mathbf{5}\rangle_B,|\mathbf{2}\rangle_A|\mathbf{5}-\mathbf{3}\rangle_B, \\
			|\mathbf{1-2}\rangle_A|\mathbf{4}\rangle_B,|\mathbf{0-1}\rangle_A|\mathbf{3}\rangle_B, \\
			|\mathbf{0+1+2}\rangle_A|\mathbf{3}+\mathbf{4}+\mathbf{5}\rangle_B
		\end{array}\right\}
		\]
		which is a strictly incompletable set. If $\mathcal{K}^A_2$ clicks, it also gives a strictly incompletable set \cite{BennettUPB1999}
		\[ 
		\left\{\begin{array}{c}
			|\mathbf{3}\rangle_A|\mathbf{4}-\mathbf{5}\rangle_B,|\mathbf{5}\rangle_A|\mathbf{5}-\mathbf{3}\rangle_B, \\
			|\mathbf{4-5}\rangle_A|\mathbf{4}\rangle_B,|\mathbf{3-4}\rangle_A|\mathbf{3}\rangle_B, \\
			|\mathbf{3+4+5}\rangle_A|\mathbf{3}+\mathbf{4}+\mathbf{5}\rangle_B
		\end{array}\right\}
		\]
		It is clear that the five post-measurement states corresponding to each outcome of Alice's measurement $\mathcal{K}^A_1,\;\mathcal{K}^A_2$ when $\mathcal{K}^B_1$ clicks form the celebrated unextendable product basis (UPB) \cite{BennettPB1999,BennettUPB1999} in $\mathbf{C^3 \otimes C^3}$. Also, the post-measurement states for each case of Alice's measurement $\mathcal{K}^A_1,\;\mathcal{K}^A_2$ when $\mathcal{K}^B_2$ clicks form the same UPB. It has been well established that a UPB is locally indistinguishable and cannot be completed, even in an extended Hilbert space \cite{Divincinzo,BennettUPB1999}. So, the set $\mathcal{S}_2$ is incompletability activable by some LOCC. Hence, this completes the proof.
	\end{proof}

    Activation of nonlocality \cite{bandhyopadhyay201} refers to the phenomenon where an initially locally distinguishable set of orthogonal states, free from local redundancy, is transformed via LOCC into a locally indistinguishable set \cite{BennettPB1999,Ghosh2001}. Activation of incompletability constitutes a stronger form of this phenomenon: an initially locally distinguishable and hence completable set \cite{BennettUPB1999}, also free from local redundancy \cite{bandhyopadhyay201}, is converted under LOCC into a strictly incompletable set. Since strict incompletability necessarily implies local indistinguishability, every instance of activation of incompletability also realizes activation of nonlocality. However, the converse does not hold in general. In this sense, activation of incompletability establishes a hierarchy within activation of nonlocality.\\

    The preceding discussion suggests that activation of incompletability is intrinsically a stronger operational phenomenon than activation of nonlocality. We now make this intuition precise by establishing the fundamental implication between the two notions. The following theorem provides a rigorous characterization of their relationship and forms the basis of the hierarchy depicted in Fig.1

\begin{theorem} Any set of orthogonal product states exhibiting activation of incompletability necessarily exhibits activation of nonlocality. However, the converse implication does not generally hold.
\begin{proof}
Let $S \subset \mathcal{H}$ be a locally distinguishable set, free from local redundancy, that is incompletability activable. Then there exists an orthogonality-preserving LOCC measurement with an outcome $\lambda$ such that the post-measurement set \[
	\mathcal{S}_\lambda = \left\{ \frac{M_\lambda^A \otimes M_\lambda^B \ket{\psi_i}}{\sqrt{p_\lambda}} \right\}
	\]
    is strictly incompletable. We claim $S_\lambda$ is locally indistinguishable, which establishes nonlocality-activation of $S$.
Suppose, for contradiction, that $S_\lambda$ is locally distinguishable. By a fundamental result of Bennett et al.\cite{BennettUPB1999}, any set of orthogonal product states that is perfectly distinguishable by LOCC is necessarily completable to a full orthonormal product basis. However, strict incompletability of $S_\lambda$ means its orthogonal complement $S_\lambda^\perp$ contains no product state, so no such completion exists \textemdash a contradiction. Therefore, $S_\lambda$ is locally indistinguishable and $S$ is nonlocality activable.\\

The implication established in Theorem 1 is therefore strict. To demonstrate that the converse fails, we now construct an explicit counterexample consisting of a locally distinguishable and completable set whose activation leads to nonlocality but not to incompletability.
\end{proof}
\end{theorem}

 In the following, we provide an example which demonstrates this separation explicitly. Consider the orthogonal set $\mathcal{S}_3$ that contains the following nine bipartite product states ${ }\left|\psi_i^{ \pm}\right\rangle, i \in \{1,2,3,4\}$ and $\left|\psi_5\right\rangle$ in $\mathbb{C}^6 \otimes \mathbb{C}^6$.
	\begin{multline}
	\begin{aligned}
		&\left|\psi_1^{ \pm}\right\rangle_{AB}=|\mathbf{0}-\mathbf{3}\rangle_A|\mathbf{0}\pm\mathbf{1}+\mathbf{4}\pm\mathbf{5} \rangle_B, \\
		&\left|\psi_2^{ \pm}\right\rangle_{AB}=|\mathbf{2}-\mathbf{5}\rangle_A|\mathbf{1}\pm\mathbf{2}+\mathbf{5}\pm\mathbf{3} \rangle_B, \\
		&\left|\psi_3^{ \pm}\right\rangle_{AB}=|\mathbf{1}\pm\mathbf{2}+\mathbf{4}\pm\mathbf{5} \rangle_A|\mathbf{0}-\mathbf{4}\rangle_B, \\
		&\left|\psi_4^{ \pm}\right\rangle_{AB}=|\mathbf{0}\pm\mathbf{1}+\mathbf{3}\pm\mathbf{4} \rangle_A|\mathbf{2}-\mathbf{3}\rangle_B, \\
		&\left|\psi_5\right\rangle_{AB}=|\mathbf{1}-\mathbf{4}\rangle_A|\mathbf{1}-\mathbf{5}\rangle_B.
	\end{aligned}
	\end{multline}
	It can be checked that the set $\mathcal{S}_3$ is not locally redundant and is distinguishable under LOCC. Hence, according to Bennett \textit{et al.,}\cite{BennettUPB1999} such a set is completable. Each party locally holds a $\mathbb{C}^2 \otimes \mathbb{C}^3$ system. For notational simplicity, we identify the local basis states as $|\mathbf{0}\rangle:=\mid 00\rangle,|\mathbf{1}\rangle:= |01\rangle,|\mathbf{2}\rangle:=|02\rangle,|\mathbf{3}\rangle:=|10\rangle,|\mathbf{4}\rangle:=|11\rangle,|\mathbf{5}\rangle:=|12\rangle$. We now show that the set $\mathcal{S}_3$ is free from local redundancy. If we consider the states $\left|\psi_1^{ +}\right\rangle_{AB}$ and $\left|\psi_1^{ -}\right\rangle_{AB},$  we clearly notice that the remaining parts of the states are non-orthogonal after discarding Bob's qubit. We can demonstrate similar results while tracing out Alice's qubit and qutrit, or Bob's qutrit part. Hence, the set $\mathcal{S}_3$ possesses no local redundancy.
    
    Moreover, the set $\mathcal{S}_3$ can be distinguished with the help of LOCC alone. The players can avail the following discrimination protocol. First Bob performs a measurement $\mathcal{M}_B \equiv\left\{\mathcal{M}_B^1:=\right.$ $P\left[|\mathbf{0}-\mathbf{4}\rangle_B\right], \mathcal{M}_B^2 :=P\left[|\mathbf{2}-\mathbf{3}\rangle_B\right], \mathcal{M}_B^3:=P[\mid\mathbf{1}-\mathbf{5}\rangle_B]$, $\mathcal{M}_B^4:=\mathbb{I}-\left(\mathcal{M}_B^1+\mathcal{M}_B^2+\mathcal{M}_B^3\right)\}$. Here, $P\left[|\cdot\rangle\right]:=$ $|\cdot\rangle\langle\cdot|_{\mathcal{P}}$, and $\mathcal{P}$ denotes the party. When $\mathcal{M}_B^1$ clicks, the given state must be $\left|\psi_3^{ +}\right\rangle_{AB}$ or $\left|\psi_3^{ -}\right\rangle_{AB}$, which can be distinguished by Alice \cite{Walgate2000}. When $\mathcal{M}_B^2$ clicks, the given state must be $\left|\psi_4^{ +}\right\rangle_{AB}$ or $\left|\psi_4^{ -}\right\rangle_{AB}$, which can be distinguished by Alice again. Also the outcome $\mathcal{M}_B^3$ isolates $\left|\psi_5\right\rangle_{AB}$. Whenever $\mathcal{M}_B^4$ clicks the given state can be $\left|\psi_1^{ +}\right\rangle_{AB}$, $\left|\psi_1^{ -}\right\rangle_{AB},$ $\left|\psi_2^{ +}\right\rangle_{AB}$ or $\left|\psi_2^{ -}\right\rangle_{AB}$. However, in that case Alice can perform a measurement $\mathcal{M}_A \equiv\left\{\mathcal{M}_A^1:=\right.$ $P\left[|\mathbf{0}\rangle_A +\left[|\mathbf{3}\rangle_A\right],\mathcal{M}_A^2:=\right.$ $P\left[|\mathbf{2}\rangle_A +\left[|\mathbf{5}\rangle_A\right] \right\}$  to distinguish between these four states. Here $\mathcal{M}_A^1$ isolates the states $\left|\psi_1^{ +}\right\rangle_{AB}$ and $\left|\psi_1^{ -}\right\rangle_{AB}$ which are distinguishable by \cite{Walgate2000}. The measurement $\mathcal{M}_A^2$ isolates the states $\left|\psi_2^{ +}\right\rangle_{AB}$ and $\left|\psi_2^{ -}\right\rangle_{AB}$, which are again distinguishable using the criteria from \cite{Walgate2000}.

    Thus, the set $\mathcal{S}_3$ is initially perfectly distinguishable by LOCC and, being complete, forms a full orthonormal basis of the underlying Hilbert space. We now show that the set $\mathcal{S}_3$ exhibits activation of nonlocality.
	
    Consider that Bob performs a local measurement $\mathcal{N}_B \equiv\left\{\mathcal{N}^B_1:=P\left[(|\mathbf{0}\rangle,|\mathbf{1}\rangle,|\mathbf{2}\rangle)_B\right], \mathcal{N}^B_2:=\right.$ $\left.P\left[(|\mathbf{3}\rangle,|\mathbf{4}\rangle,|\mathbf{5}\rangle)_B\right]\right\}$, $P\left[(|i\rangle,|j\rangle,\dots)_{\mathcal{P}}\right] = \left[(|i\rangle\langle i|+|j\rangle\langle j|+\dots)_{\mathcal{P}}\right]$, $\mathcal{P}$ stands for party. If $\mathcal{N}^B_1$ clicks, they end up with
		
		$$
		\left\{\begin{array}{c}
			|\mathbf{0-3}\rangle_A|\mathbf{0}\pm\mathbf{1}\rangle_B,|\mathbf{2-5}\rangle_A|\mathbf{1}\pm\mathbf{2}\rangle_B, \\
			|\mathbf{1\pm2+4\pm5}\rangle_A|\mathbf{0}\rangle_B,|\mathbf{0\pm1+3\pm4}\rangle_A|\mathbf{2}\rangle_B, \\
			|\mathbf{1-4}\rangle_A|\mathbf{1}\rangle_B\\
			
		\end{array}\right\}
		$$

After that Alice makes measurement  $\mathcal{N}_A \equiv\left\{\mathcal{N}^A_1:=P\left[(|\mathbf{0}\rangle,|\mathbf{1}\rangle,|\mathbf{2}\rangle)_A\right], \mathcal{N}^A_2:=\right.$ $\left.P\left[(|\mathbf{3}\rangle,|\mathbf{4}\rangle,|\mathbf{5}\rangle)_A\right]\right\}$. If $\mathcal{N}^A_1$ occurs, the states reduce to
		\[
		\left\{\begin{array}{c}
			|\mathbf{0}\rangle_A|\mathbf{0}\pm\mathbf{1}\rangle_B,|\mathbf{2}\rangle_A|\mathbf{1}\pm\mathbf{2}\rangle_B, \\
			|\mathbf{1\pm2}\rangle_A|\mathbf{0}\rangle_B,|\mathbf{0\pm1}\rangle_A|\mathbf{2}\rangle_B, \\
			|\mathbf{1}\rangle_A|\mathbf{1}\rangle_B
		\end{array}\right\}
		\]
		which is a strictly indistinguishable set \cite{BennettPB1999, Walgate2002}. If $\mathcal{N}^A_2$ occurs, the resulting set is also strictly indistinguishable \cite{BennettPB1999, Walgate2002}.
		\[ 
		\left\{\begin{array}{c}
			|\mathbf{3}\rangle_A|\mathbf{0}\pm\mathbf{1}\rangle_B,|\mathbf{5}\rangle_A|\mathbf{1}\pm\mathbf{2}\rangle_B, \\
			|\mathbf{4\pm5}\rangle_A|\mathbf{0}\rangle_B,|\mathbf{3\pm4}\rangle_A|\mathbf{2}\rangle_B, \\
			|\mathbf{4}\rangle_A|\mathbf{1}\rangle_B
		\end{array}\right\}
		\]

On the other hand, if Bob gets $\mathcal{N}^B_2$, they are then left with the following states
		\[	\left\{\begin{array}{c}
			|\mathbf{0-3}\rangle_A|\mathbf{4}\pm\mathbf{5}\rangle_B,|\mathbf{5}\rangle_A|\mathbf{5}\pm\mathbf{3}\rangle_B, \\
			|\mathbf{1\pm2+4\pm5}\rangle_A|\mathbf{4}\rangle_B,|\mathbf{0\pm1+3\pm4}\rangle_A|\mathbf{3}\rangle_B, \\
			|\mathbf{1-4}\rangle_A|\mathbf{5}\rangle_B\\

		\end{array}\right\}
		\]
		After that Alice makes measurement $\mathcal{N}^A_1$ or $\mathcal{N}^A_2$.  If $\mathcal{N}^A_1$ clicks, it gives
		\[ 
		\left\{\begin{array}{c}
			|\mathbf{0}\rangle_A|\mathbf{4}\pm\mathbf{5}\rangle_B,|\mathbf{2}\rangle_A|\mathbf{5}\pm\mathbf{3}\rangle_B, \\
			|\mathbf{1\pm2}\rangle_A|\mathbf{4}\rangle_B,|\mathbf{0\pm1}\rangle_A|\mathbf{3}\rangle_B, \\
			|\mathbf{1}\rangle_A|\mathbf{5}\rangle_B
		\end{array}\right\}
		\]
		which is a strictly indistinguishable set \cite{BennettPB1999, Walgate2002}. If $\mathcal{N}^A_2$ clicks, it also gives a strictly indistinguishable set \cite{BennettPB1999, Walgate2002}.
		\[ 
		\left\{\begin{array}{c}
			|\mathbf{3}\rangle_A|\mathbf{4}\pm\mathbf{5}\rangle_B,|\mathbf{5}\rangle_A|\mathbf{5}\pm\mathbf{3}\rangle_B, \\
			|\mathbf{4\pm5}\rangle_A|\mathbf{4}\rangle_B,|\mathbf{3\pm4}\rangle_A|\mathbf{3}\rangle_B, \\
			|\mathbf{4}\rangle_A|\mathbf{5}\rangle_B
		\end{array}\right\}.
		\]

As illustrated by the previous example, the given set displays the activation of nonlocality. In accordance with Bennett et al. \cite{BennettPB1999}, the four sets of nine states constitute a complete basis within the Hilbert space $\mathbb{C}^3 \otimes \mathbb{C}^3$, which explicitly proves their completability.

  Hence, the set $\mathcal{S}_3$ provides an example of activation of nonlocality (purple region) that lies strictly outside the class of sets exhibiting activation of incompletability (yellow region); see Fig.~\ref{fig1}.
   To obtain a deeper understanding of the activation of incompletability, we now consider a more restricted class of local operations, namely local incoherent operations and classical communication (LICC). Since LICC forms a proper subclass of LOCC, it provides a natural framework for investigating the interplay between incompletability and quantum coherence. In the next section, we derive several fundamental results concerning activation of incompletability within the LICC paradigm.
   
	\section{Distinguishability by LICC}
    \label{A4}
A resource theory is defined by two fundamental ingredients: a set of free states and a class of free operations \cite{c2, c3}. Free states are those that contain no resource and can be prepared at no cost under the physical restrictions of the theory. Free operations are physical processes that, by definition, cannot generate a resource state from a free state.

In the resource theory of quantum coherence \cite{T2014,W2016,S2017}, the free states are the incoherent states, i.e., density matrices that are diagonal in a fixed reference basis. Correspondingly, free operations are those completely positive trace-preserving (CPTP) maps that cannot create coherence from incoherent states.

 A CPTP map $\Lambda(\rho) = \sum_n K_n \rho K_n^\dagger$ is called an incoherent operation \cite{c1} if it admits a Kraus decomposition 
 $\{K_n\}$ such that each Kraus operator maps the set of incoherent states 
 $\mathcal{I}$ into itself. \[K_n \mathcal{I} K_n^\dagger \subset \mathcal{I} \quad \forall n\]Equivalently, this means that for each $n$, the operator $K_n$ must satisfy:\[\frac{K_n \rho K_n^\dagger}{\text{Tr}(K_n \rho K_n^\dagger)} \in \mathcal{I}, \quad \forall \rho \in \mathcal{I}\]
When moving to distributed scenarios involving spatially separated parties, additional constraints arise. If each party is restricted to performing only local incoherent operations and they are allowed to communicate solely through classical channels, the resulting class of operations is known as LICC \cite{c1}. Since incoherent operations are more restrictive than general quantum operations, we naturally obtain
\[\text{LICC} \subset \text{LOCC}\]
This operational restriction provides the bridge to entanglement. Under LOCC, entanglement is the resource that cannot be created. Under LICC, coherence is additionally constrained at the local level. Remarkably, these two resource theories become deeply connected. To proceed further, we first introduce several basic concepts and definitions that will be used in the remainder of this work.
	We will require the concepts of von Neumann entropy and relative entropy between quantum states. The von Neumann entropy of a quantum state $\varrho$ is given by \cite{quantum entanglement}
	\begin{equation*}
		S(\varrho)=-\operatorname{tr}\left(\varrho \log _2 \varrho\right) 
	\end{equation*}
	The von Neumann relative entropy between two quantum states, $\varrho$ and $\zeta$ is given by \cite{quantum entanglement}
	\begin{equation*}
		S(\varrho \| \zeta)=\operatorname{tr}\left(\varrho \log _2 \varrho-\varrho \log _2 \zeta\right) 
	\end{equation*}
	It is to be noted that the relative entropy is not symmetric with respect to its arguments.
	A qualitative definition of quantum coherence, as has already been given in the literature, can be as follows.
	\begin{defin}
	  A pure quantum state $|\psi\rangle$ of a physical system represented by a Hilbert space $\mathbb{C}^d$ is said to be quantum coherent with respect to a reference orthonormal basis say $\left\{\ket{i}\right\}_{i=0}^{d-1}$ if it cannot be written as a single basic vector $\ket{i}$.
	The notion has also been quantified, and one of the quantifications is as follows.
	\end{defin}
	\begin{defin} \cite{T2014}
	 Let $B$ be a complete orthonormal basis of pure states in $\mathbb{C}^d$. Let $C_B(|\psi\rangle)$ be the relative entropy of quantum coherence of $|\psi\rangle \in \mathbb{C}^d$, so
	
	\begin{equation*}
		C_B(|\psi\rangle)=\min _{\rho_B \in M_s} S\left(|\psi\rangle\langle\psi| \| \rho_B\right), \tag{3}
	\end{equation*}
	where $M_B$ is the set of all probabilistic mixtures of the projectors onto the elements of $B$.
	\end{defin}
	Complete orthonormal bases of bipartite quantum systems are, of course, distinguishable under global operations. One just makes a measurement onto that basis. Things are more complicated, however, when a restricted class of operations is allowed, for example, LOCC. If a complete orthonormal basis is also distinguishable under LOCC, we will call the basis locally distinguishable.
	
	A bipartite pure state is said to be entangled if it cannot be written as a tensor product of pure states of the two systems. Interestingly, the following result establishes a direct connection between quantum coherence and entanglement in the context of locally distinguishable bases.
	\begin{prop} \cite{A2021}
		Let $|\psi\rangle \in \mathcal{H}_A \otimes \mathcal{H}_B$ be a bipartite pure state. Then $|\psi\rangle$ is entangled if and only if, for every complete orthonormal basis of $\mathcal{H}_A \otimes \mathcal{H}_B$ that is perfectly distinguishable by LOCC, $|\psi\rangle$ has nonzero quantum coherence with respect to that basis.
	\end{prop}
\begin{proof}
	 Let $|\psi\rangle$ be a pure entangled state of a bipartite quantum system, the Hilbert space corresponding to which is $\mathbb{C}^{d_1} \otimes \mathbb{C}^{d_2}$. Let us now consider an arbitrary locally distinguishable complete orthonormal basis. $|\psi\rangle$ will have vanishing quantum coherence in this basis if and only if it is an element of this basis. However, if $|\psi\rangle$ is an element of this basis, the latter cannot be locally distinguishable, as was proven in Ref. \cite{Horodecki2003} that any complete orthonormal basis containing even a single entangled state cannot be locally distinguishable. Therefore, $|\psi\rangle$ must have a nonzero quantum coherence in any locally distinguishable complete orthonormal basis.
	
	On the other hand, an arbitrary pure product state can always be expanded to form a complete biorthonormal product basis, which can always be distinguished by LOCC. Therefore, the product state has zero quantum coherence, at least with respect to this basis.
\end{proof}

	So, entangled quantum states can be interpreted as manifestations of quantum coherence when examined in locally distinguishable bases. In particular, a multipartite pure state is entangled across at least one bipartition if and only if it remains quantum coherent with respect to every locally distinguishable complete orthonormal basis.

For general bipartite states, possibly mixed, this connection becomes quantitative. The minimal relative entropy of coherence, optimized over all locally distinguishable bases, is lower bounded by the relative entropy of entanglement. Therefore, entanglement emerges as a constrained form of quantum coherence - specifically, the portion of coherence that persists under all locally distinguishable classical descriptions.

We now investigate the role of LICC in distinguishing between activation of nonlocality and activation of incompletability. In particular, LICC imposes fundamental structural constraints on sets whose incompletability can be activated within this restricted operational framework. The following proposition addresses this question.
\begin{theorem}
	Let $\mathcal{S} \subset \mathcal{H}$ be a set of mutually orthogonal product states that is perfectly distinguishable via LOCC, and suppose that $\mathcal{S}$ exhibits incompletability activable by LICC. Then there exists a completion $\widetilde{\mathcal{S}}$ of $\mathcal{S}$ to an orthonormal basis of $\mathcal{H}$ which is not perfectly distinguishable by LOCC.
\end{theorem}
\begin{proof}
	Assume that the initial set $\mathcal{S}=\{\ket{\psi_i}\}$ consists of mutually orthogonal product states that are perfectly distinguishable by LOCC and is also completable to a locally distinguishable orthonormal basis.
	
	By Proposition~4, a bipartite pure state is entangled if and only if it exhibits nonzero quantum coherence with respect to every locally distinguishable complete orthonormal basis \cite{A2021}. Since each $\ket{\psi_i}$ is a product state and $\mathcal{S}$ is distinguishable, there exists a locally distinguishable basis with respect to which all $\ket{\psi_i}$ have zero coherence. Hence, the average coherence (and average entanglement) of the set $\mathcal{S}$ is zero.
	
	Now consider an LICC protocol applied to $\mathcal{S}$. By the assumption of incompletability activation, there exists an outcome $\lambda$ (occurring with nonzero probability $p_\lambda > 0$) such that the post-measurement set
	\[
	\mathcal{S}_\lambda = \left\{ \frac{M_\lambda^A \otimes M_\lambda^B \ket{\psi_i}}{\sqrt{p_\lambda}} \right\}
	\]
	is strictly incompletable.
	
	Strict incompletability implies that there exists no complete orthonormal basis, perfectly distinguishable by LOCC, that contains $\mathcal{S}_\lambda$. Consequently, by Proposition 4, each state in $\mathcal{S}_\lambda$ must exhibit nonzero coherence with respect to every locally distinguishable basis. Therefore, the average coherence of the post-measurement set $\mathcal{S}_\lambda$ is strictly positive.
	This leads to a contradiction, since LICC (being a class of local incoherent operations assisted by classical communication) cannot increase the average coherence of a set of states. 
	
	Hence, the assumption is false, and the result follows.
\end{proof}
	\section{Conclusion}
\label{A5}
In this work, we have introduced the notion of \emph{activation of incompletability} in the context of local quantum state discrimination. We showed that there exist sets of orthogonal states that are initially perfectly distinguishable by LOCC and free from local redundancy, yet can be transformed through LOCC into strictly incompletable sets. This demonstrates that incompletability, much like local indistinguishability, can emerge as an activated nonclassical feature under local measurements.

We established that activation of incompletability constitutes a strictly stronger phenomenon than activation of nonlocality: every set that exhibits activation of incompletability necessarily exhibits activation of nonlocality, while the converse does not hold in general. This reveals a nontrivial hierarchy within activation phenomena associated with local state discrimination. 

Furthermore, by analyzing the problem within the framework of LICC, we uncovered a fundamental connection between incompletability, coherence, and entanglement. In particular, we proved that any set whose incompletability can be activated under LICC can be completed to a full orthonormal basis, although the resulting completed basis ceases to be perfectly distinguishable by LOCC. Our results, therefore, highlight an intrinsic interplay between coherence-theoretic restrictions and locally accessible quantum information.

These findings provide a new operational perspective on incompletability and suggest several open directions, including the characterization of incompletability activation in multipartite systems, its relation to other resource-theoretic notions of nonclassicality, and the identification of minimal-dimensional examples exhibiting stronger activation structures.

\section{ACKNOWLEDGMENT}
A. Bhunia, acknowledges the support of ANRF-NPDF (File no. PDF/2025/002762).

\end{document}